\documentclass[11pt]{article}
\usepackage{amsfonts, amsmath, amsthm, amssymb, mathtools, enumitem, tikz, hyperref, caption, accents, dsfont, bbm}
\usepackage{eucal}
\usepackage[noabbrev,capitalize]{cleveref}

\definecolor{darkgray}{RGB}{64,64,64}
\definecolor{litegray}{RGB}{192,192,192}
\definecolor{green}{HTML}{0F9D58}
\definecolor{red}{HTML}{DB4437}
\tikzstyle{vertex}=[circle, draw, fill=red, inner sep=0pt, minimum width=5pt]
\tikzstyle{vtx}=[circle, draw, fill=litegray, inner sep=0pt, minimum width=5pt]

\usepackage{microtype, fullpage}
\newtheorem{theorem}{Theorem}[section]

\newtheorem{proposition}[theorem]{Proposition}
\newtheorem{corollary}[theorem]{Corollary}

\newtheorem*{minimax}{The minimax theorem (von Neumann~\cite{vN})}

\theoremstyle{definition}
\newtheorem{definition}[theorem]{Definition}
\newtheorem{example}[theorem]{Example}

\newtheorem{remark}[theorem]{Remark}

\theoremstyle{remark}

\newtheorem*{claim*}{Claim}

\newcommand{\p}{\mathbf{p}}
\newcommand{\q}{\mathbf{q}}
\newcommand{\z}{\mathbf{z}}
\newcommand{\bb}{\mathbf{b}}
\newcommand{\rr}{\mathbf{r}}
\newcommand{\cc}{\mathbf{c}}

\newcommand{\mix}{\mathrm{mix}}
\newcommand{\down}{\mathrm{down}}
\newcommand{\cl}{\mathrm{cl}}

\newcommand{\BB}{\mathcal{B}}
\newcommand{\FF}{\mathcal{F}}

\newcommand{\EE}{\mathcal{E}}

\newcommand{\LL}{\mathcal{L}}
\newcommand{\MM}{\mathcal{M}}

\newcommand{\XX}{\mathcal{X}}
\newcommand{\R}{\mathbb{R}}
\newcommand{\Q}{\mathbb{Q}}

\newcommand{\E}{\mathbb{E}}
\newcommand{\N}{\mathbb{N}}

\newcommand{\ubar}[1]{\underaccent{\bar}{#1}}
\newcommand{\ind}{\mathbbm{1}}

\title{The minimax property in infinite two-person win-lose games}

\author{
  Ron Holzman\footnote{Department of Mathematics, Technion -- Israel Institute of Technology, Technion City, Haifa 3200003, Israel. Email: {\tt holzman@technion.ac.il}.}
}

\date{}

\begin{document}

\maketitle

\begin{abstract}
We explore a version of the minimax theorem for two-person win-lose games with infinitely many pure strategies. In the countable case, we give a combinatorial condition on the game which implies the minimax property. In the general case, we prove that a game satisfies the minimax property along with all its subgames if and only if none of its subgames is isomorphic to the ``larger number game." This generalizes a recent theorem of Hanneke, Livni and Moran. We also propose several applications of our results outside of game theory.

\vspace{10pt}

\end{abstract}

\section{Introduction} \label{sec:introduction}

The minimax theorem is historically considered to be the starting point of game theory, and plays a fundamental role in it. Its standard formulation is for \emph{finite two-person zero-sum games}. In such a game, written as $G = (S,T,\pi)$, players $1$ and $2$ simultaneously choose their pure strategies $s \in S$ and $t \in T$, respectively, where $S$ and $T$ are finite, non-empty sets. The function $\pi: S \times T \to \R$ determines the payoffs: $\pi(s,t)$ to player $1$ and $-\pi(s,t)$ to player $2$. In the \emph{mixed extension} $G^{\mix} = (\Delta(S),\Delta(T),\pi^{\mix})$ of the game $G$, players $1$ and $2$ simultaneously choose their mixed strategies $\p \in \Delta(S)$ and $\q \in \Delta(T)$, respectively, where $\Delta(X)$ is the set of probability distributions over $X$. The payoff function $\pi^{\mix}: \Delta(S) \times \Delta(T) \to \R$ extends $\pi$ by taking the expectation: $\pi^{\mix}(\p,\q) = \E_{(s,t) \sim \p \times \q} \pi(s,t)$.

\begin{minimax}
Let $G = (S,T,\pi)$ be a finite two-person zero-sum game. Then in its mixed extension $G^{\mix} = (\Delta(S),\Delta(T),\pi^{\mix})$ we have
\begin{equation} \label{eq1}
\max_{\p \in \Delta(S)} \min_{\q \in \Delta(T)} \pi^{\mix}(\p,\q) = \min_{\q \in \Delta(T)} \max_{\p \in \Delta(S)} \pi^{\mix}(\p,\q).
\end{equation}
\end{minimax}

The left-hand side of equation~(\ref{eq1}) is what player $1$ can expect to get in a scenario where player $1$ first commits to a mixed strategy $\p$ and then player $2$, knowing $\p$ (but not its realization) chooses a best response to it. The right-hand side is what player $1$ can expect to get in the opposite scenario, where player $2$ first commits to a mixed strategy $\q$ and then player $1$ best responds to it. In the pure-strategy setting, the two scenarios typically lead to different outcomes: knowing the strategy of one's opponent before choosing one's own strategy is a clear advantage. The significance of the minimax theorem is that it shows that this advantage disappears when the players have access to mixed strategies (i.e., flipping coins to choose their pure strategy). In the mixed extension, the two opposite scenarios lead to the same expected payoff $v$ to player $1$, namely the common value of the two sides of equation~(\ref{eq1}). This $v$ is called the \emph{value} of $G$ in mixed strategies. It is a payoff level $v$ so that each player can unilaterally guarantee that the expected outcome will be (from that player's perspective) at least as good as $v$, regardless of the opponent's choice of strategy.

Actually, von Neumann proved a more general version of the minimax theorem: while the formulation above looks at a bilinear function $\pi^{\mix}: \Delta(S) \times \Delta(T) \to \R$, he considered a function $f: X \times Y \to \R$ where $X$ and $Y$ are compact and convex subsets of Euclidean spaces, and $f$ is continuous, concave in the first variable and convex in the second. There is a large body of literature extending the minimax theorem in this form, by allowing $X$ and $Y$ to be subsets of infinite-dimensional spaces and/or relaxing some of the compactness, continuity or concavity/convexity conditions. See, e.g., Fan~\cite{fan}, Sion~\cite{sion} or the survey by Simons~\cite{simons}. Another line of literature, exemplified by Yanovskaya~\cite{yan}, has considered the mixed extension of games with infinitely many pure strategies, allowing the use of mixed strategies which are finitely additive but not necessarily $\sigma$-additive.

In this paper, we consider games with infinitely many pure strategies, in which every outcome is a win for one player and a loss for the other. Focusing on win-lose games will allow us to find conditions for the minimax theorem to hold, which are based on the combinatorial structure of the payoff function, rather than on continuity and concavity/convexity properties as in most of the literature mentioned above. We will, however, briefly comment on how our results may be extended from win-lose games to real-valued games.

\begin{definition} \label{def:11}
A \emph{two-person win-lose game} is a tuple $G = (S,T,\pi)$, where:
\begin{itemize}
\item $S$ is the set of pure strategies of player $1$,
\item $T$ is the set of pure strategies of player $2$,
\item $\pi: S \times T \to \{0,1\}$ is the payoff function.
\end{itemize}
The sets $S,T$ are non-empty and may be infinite. The game $G$ is \emph{countable} if $|S| = |T| = \aleph_0$.
\end{definition}

We will use the term \emph{game} for brevity, to mean a two-person win-lose game. Note that the payoff $1$ stands for a win of player $1$, while the payoff $0$ stands for a win of player $2$ (a choice of $\pm 1$ for these two outcomes would be more symmetric, but less convenient). As is standard for finite games, a game $G$ can be represented by a $0/1$-matrix $M_G$ with rows indexed by $s \in S$, columns indexed by $t \in T$, and $\pi(s,t)$ appearing where row $s$ and column $t$ meet; but such a matrix may have infinitely many rows and columns.

In order to speak about the mixed extension of a game $G$, we need to specify the class of mixed strategies that are allowed. We will work mostly with countable games. In this case, we will naturally call a \emph{mixed strategy} of player $1$ any assignment of probabilities $\p = (p_s)_{s \in S}$ to the pure strategies (equivalently, to the rows), such that $p_s \ge 0$ and $\sum_{s \in S} p_s = 1$, and similarly for player $2$. When $S$ (or $T$) is uncountable, we keep the same definition of a mixed strategy, with the understanding that $p_s > 0$ for at most countably many $s \in S$. (We could allow more general probability distributions over $S$ or $T$, but then the definition would need to cope with measurability issues, and our results would not be affected anyway.)

Thus, for a set $X$ of pure strategies, we let the set $\Delta(X)$ of mixed strategies consist of the probability distributions over $X$ having at most countable support. Given mixed strategies $\p \in \Delta(S)$ and $\q \in \Delta(T)$, we define $\pi^{\mix}(\p,\q) = \E_{(s,t) \sim \p \times \q} \pi(s,t)$. Note that this is the probability that player $1$ wins under $\p \times \q$.

\begin{definition} \label{def:12}
We say that the game $G = (S,T,\pi)$ satisfies the \emph{minimax property} if the following equality holds:
\begin{equation} \label{eq2}
\sup_{\p \in \Delta(S)} \inf_{\q \in \Delta(T)} \pi^{\mix}(\p,\q) = \inf_{\q \in \Delta(T)} \sup_{\p \in \Delta(S)} \pi^{\mix}(\p,\q).
\end{equation}
\end{definition}

Equation~(\ref{eq2}) is the same as equation~(\ref{eq1}), except that we do not insist that the maximum or minimum be attained. For reasons of tradition, we still call this the minimax (rather than the infisup) property. When we can actually assert this equality with $\max$ instead of $\sup$ or $\min$ instead of $\inf$, we will do so. Note that the left-hand side of equation~(\ref{eq2}) is always less than or equal to the right-hand side. Equality means that the game $G$ has an \emph{$\varepsilon$-value} $v$ in mixed strategies, namely, each player can unilaterally guarantee that the expected outcome will be, up to any $\varepsilon > 0$, at least as good as $v$ (from that player's perspective).

Not all games satisfy the minimax property. A standard counterexample is the following.

\begin{example} \label{ex:13}
The \emph{larger number game (LNG)}.
\\In this game, each of the two players chooses a positive integer, and the one who chose the larger number wins. Specifying how ties are broken is immaterial for the analysis, but for concreteness we will assume that they are broken in favor of player $1$. Thus, $G = (\N,\N,\pi)$ where
$$\pi(s,t) = \begin{cases} 1 \quad  & \text{if $s \ge t$},\\ 0 \quad & \text{otherwise}. \end{cases}$$
This game is represented by the following matrix $M_G$:

\vspace{-10pt}

\begin{center}
\begin{tabular}{c | c c c c c c}
& $1$ & $2$ & $3$ & $4$ & $5$ & $\cdots$ \\
\hline
$1$ & $1$ & $0$ & $0$ & $0$ & $0$ & \\
$2$ & $1$ & $1$ & $0$ & $0$ & $0$ & \\
$3$ & $1$ & $1$ & $1$ & $0$ & $0$ & $\cdots$ \\
$4$ & $1$ & $1$ & $1$ & $1$ & $0$ & \\
$5$ & $1$ & $1$ & $1$ & $1$ & $1$ & \\
$\vdots$ & & & $\vdots$ & & & $\ddots$
\end{tabular}
\end{center}

\vspace{10pt}

Intuitively, this game violates the minimax property because it is clearly advantageous to know your opponent's strategy before choosing your own. Formally, the left-hand side of equation~(\ref{eq2}) equals $0$ in this example, because given any $\p = (p_s)_{s \in \N} \in \Delta(\N)$ and $\varepsilon > 0$, player $2$ can hold the winning probability of player $1$ below $\varepsilon$ by choosing a number $t \in \N$ for which $\sum_{s=1}^{t-1} p_s > 1 - \varepsilon$. A symmetric argument shows that the right-hand side of equation~(\ref{eq2}) equals $1$. The gap between the values of the two sides is as large as it can possibly be.
\end{example}

In view of the above example, our goal in this paper is to identify combinatorial conditions on the game $G$ under which the minimax property does hold. Our first main result gives a sufficient condition for a countable game $G$ to satisfy the minimax property. To formulate the condition, we first extend the game $G = (S,T,\pi)$ by adding to the matrix $M_G$ all possible new rows which are less than some old row (in the natural partial order on $\{0,1\}^T$). This adds new pure strategies for player $1$, each of which is weakly dominated by some strategy in $S$. Hence this extension is innocuous in regard to what the players can achieve in the game. Our sufficient condition for $G$ to have the minimax property is: for any infinite sequence $s_1,s_2, \ldots$ of player $1$'s strategies in the extended game, in which $s_i$ weakly dominates $s_j$ for all $i > j$, there is a strategy $s$ that weakly dominates every $s_i$. Note that the LNG game (Example~\ref{ex:13}) violates this condition for the strategy sequence $1,2,\ldots$. In fact, to satisfy the condition, player $1$ would need to have $\infty$ as well in the strategy set. The formal statement of this result, its proof and some consequences appear in Section~\ref{sec:2}.

The above condition is sufficient but not necessary for the game $G$ to satisfy the minimax property. One reason for that is its asymmetric treatment of the two players. Another reason, which pertains also to other conceivable conditions that forbid a certain substructure in the matrix $M_G$, is that a forbidden substructure may appear in one part of the matrix, but $G$ may still satisfy the minimax property due to another part of the matrix. Therefore, it makes more sense to try to characterize those games $G$ which intrinsically satisfy the minimax property, in the sense of the following definition.

\begin{definition} \label{def:14}
Let $G = (S,T,\pi)$ be a game. A \emph{subgame} of $G$ is a game of the form $G' = (S',T',\pi ')$, where $S' \subseteq S, T' \subseteq T$, and $\pi '$ is the restriction of $\pi$ to $S' \times T'$. We say that the game $G$ is \emph{totally minimax} if every subgame $G'$ of $G$ satisfies the minimax property.
\end{definition}

Clearly, for $G$ to be totally minimax, it must not contain LNG (or any game isomorphic to it) as a subgame. Our second main result shows that the absence of any LNG subgame is also a sufficient condition for $G$ to be totally minimax. (This applies to all, not just countable games.) In other words, when $G$ is not totally minimax, this can always be attributed to a simple reason, namely the presence of an LNG subgame. This strengthens a recent result of Hanneke, Livni and Moran~\cite{hlm}, who proved the same under an additional condition on the game $G$: they required that the VC-dimension of the matrix $M_G$ be finite (the VC-dimension is a fundamental complexity measure in statistical learning theory). The characterization result for totally minimax games is formulated and proved in Section~\ref{sec:3}.

The final Section~\ref{sec:4} describes connections and applications of our results to several topics outside of game theory. These include convergence of number series, row and column densities in infinite binary matrices, fractional matchings and covers in infinite hypergraphs, stability in first-order logic, and online learning theory.

\section{A sufficient condition for the minimax property in countable games} \label{sec:2}
Let $G = (S,T,\pi)$ be a game. We assume w.l.o.g.\ that the rows in the matrix $M_G$ are all distinct, i.e., no two strategies of player $1$ are equivalent. Thus, each pure strategy $s \in S$ corresponds bijectively to its row $(\pi(s,t))_{t \in T}$ in $\{0,1\}^T$, or equivalently to the subset $B_s$ of $T$ defined by \[ B_s = \{ t \in T: \pi(s,t) = 1\}.\] In words, for a strategy $s$ of player $1$, $B_s$ is the set of player $2$'s strategies that $s$ beats. Thus, the containment relation $B_s \subset B_{s'}$ corresponds to the game theoretic property that $s$ is weakly dominated by $s'$. We denote the family of all these sets by $\BB^1(G)$, i.e., \[\BB^1(G) = \{B_s: s \in S\}.\] Clearly, the set $T$ and the family of subsets $\BB^1(G)$ fully describe the game $G$. For example, when $G$ is the larger number game, $\BB^1(G)$ is the family $\{\{1\},\{1,2\},\ldots\}$. We also consider the downward closure $\BB^1_{\down}(G)$ of the family $\BB^1(G)$, namely the family \[\BB^1_{\down}(G) = \{ A \subseteq T: \exists s \in S \,\,\text{such that}\, A \subseteq B_s\}.\]

\noindent The family $\BB^1_{\down}(G)$ corresponds to an extension of the game $G$, obtained by adding to player $1$'s strategy set all possible new strategies that are weakly dominated by some old strategy. This extension has no effect on what the players can achieve in the game, but it facilitates the statement of our sufficient condition for the minimax property. This condition is based on the following set theoretic property.

\begin{definition} \label{def:21}
Let $\FF$ be a family of sets. We say that $\FF$ is \emph{ascending-union closed} if \[A_i \in \FF, i = 1,2,\ldots \,\,\text{and}\,\, A_1 \subset A_2 \subset \cdots \subset A_i \subset \cdots \,\, \Longrightarrow \,\, \cup_{i=1}^{\infty} A_i \in \FF.\]
\end{definition}

In words, the definition requires that the family $\FF$ be closed with respect to taking unions of countable ascending chains. Our sufficient condition for the minimax property states that the family $\BB^1_{\down}(G)$ should be ascending-union closed. Recalling the correspondence between containment and domination, we can interpret the condition as follows: given any infinite ascending (in the sense of weak domination) sequence of pure strategies of player $1$ in the extended game, there should be a strategy which weakly dominates all of them.

The main result of this section is the following.

\begin{theorem} \label{thm:22}
Let $G = (S,T,\pi)$ be a countable game, such that the family $\BB^1_{\down}(G)$ is ascending-union closed. Then $G$ satisfies the minimax property in the following form:
\begin{equation} \label{eq3}
\max_{\p \in \Delta(S)} \inf_{\q \in \Delta(T)} \pi^{\mix}(\p,\q) = \inf_{\q \in \Delta(T)} \max_{\p \in \Delta(S)} \pi^{\mix}(\p,\q).
\end{equation}
\end{theorem}

\begin{proof} By identifying $T$ with the set $\N$ of positive integers, and each subset $A$ with its characteristic function $\ind
_A: \N \to \{0,1\}$, we view $\BB^1_{\down}(G)$ as a subset of $\{0,1\}^{\N}$. Henceforth, we move freely between sets and their characteristic functions. Equipping $\{0,1\}$ with the discrete topology and $\{0,1\}^{\N}$ with the corresponding product topology, the latter is a compact topological space by Tychonoff's theorem. Moreover, as $\N$ is countable, the space $\{0,1\}^{\N}$ is metrizable. Treating $\BB^1_{\down}(G)$ as a subset of the compact metric space $\{0,1\}^{\N}$ is the key idea of our proof, together with the observation made in the first step below.

\textbf{Step 1} $\BB^1_{\down}(G)$ is a closed subset of $\{0,1\}^{\N}$, hence is itself a compact metric space.

\noindent \emph{Proof.} Let $\z^{(i)} = (z_1^{(i)},z_2^{(i)},\ldots) \in \BB^1_{\down}(G)$, $i = 1,2,\ldots$, and assume that the sequence $\z^{(i)}$ converges to $\z = (z_1,z_2,\ldots) \in \{0,1\}^{\N}$ in the product topology (i.e., pointwise). We have to show that $\z \in \BB^1_{\down}(G)$. Let $A = \{ j \in \N: z_j = 1\}$. If $A$ is finite then for large enough $i$ we have $A \subseteq \{j \in \N: z_j^{(i)} = 1\}$. Since $\z^{(i)}$ is in $\BB^1_{\down}(G)$, so is $\z$, and we are done. Thus we may assume that $A$ is an infinite subset of $\N$, and view it as the ascending union of $A_1 \subset A_2 \subset \cdots \subset A_i \subset \cdots$, where $A_i$ consists of the first $i$ elements of $A$. Repeating the argument just made, we see that each $A_i$ is in $\BB^1_{\down}(G)$. Since $\BB^1_{\down}(G)$ is ascending-union closed, it follows that $A \in \BB^1_{\down}(G)$, as required.

\textbf{Step 2} Let $G_n = (S,T_n,\pi_n)$ be the subgame of $G$ obtained by restricting player $2$ to the first $n$ pure strategies. Then for every $n$, the game $G_n$ has a value $v_n$, and the sequence $(v_n)_{n=1,2,\ldots}$ is non-increasing and converges to some $v \in [0,1]$.

\noindent \emph{Proof.} In $G_n$ player $1$ has at most $2^n$ pure strategies, up to equivalence. Thus $G_n$ is equivalent to a finite game, so by the finite minimax theorem it has a value $v_n \in [0,1]$. Clearly, the sequence $(v_n)_{n=1,2,\ldots}$ is non-increasing, and therefore it converges to some $v \in [0,1]$.

In the following two steps, we show that the two sides of equation~(\ref{eq3}) are well defined, and both are equal to the $v$ of Step~2.

\textbf{Step 3} For each $\q \in \Delta(T)$, $\max_{\p \in \Delta(S)} \pi^{\mix}(\p,\q)$ is attained at some pure strategy in $S$, and we have $\inf_{\q \in \Delta(T)} \max_{\p \in \Delta(S)} \pi^{\mix}(\p,\q) \le v$.

\noindent \emph{Proof.} Fix $\q = (q_j)_{j \in \N} \in \Delta(T)$, and consider the mapping $f: \BB^1_{\down}(G) \to [0,1]$ which maps $\z = (z_j)_{j \in \N}$ to $\sum_{j \in \N} q_j z_j$. It is easy to check that $f$ is continuous with respect to the product topology. By Step~1, $\BB^1_{\down}(G)$ is compact, hence $f$ attains its maximum at some $\z \in \BB^1_{\down}(G)$. From the way that $\BB^1_{\down}(G)$ was derived from $\BB^1(G)$, the maximum of $f$ over $\BB^1(G)$ is the same. This means that there is a pure strategy $s \in S$ which maximizes $\pi^{\mix}(s',\q)$ over all $s' \in S$. Hence this $s \in S$, viewed as a mixed strategy, also maximizes $\pi^{\mix}(\p,\q)$ over all $\p \in \Delta(S)$.

To prove the second part of the statement of Step~3, consider for each $n$ an optimal mixed strategy $\q_n \in \Delta(T_n)$ of player $2$ in $G_n$. This $\q_n$ can also be viewed as a mixed strategy of player $2$ in $G$, and as such it satisfies $\max_{\p \in \Delta(S)} \pi^{\mix}(\p,\q_n) = v_n$. Thus we have \[ \inf_{\q \in \Delta(T)} \max_{\p \in \Delta(S)} \pi^{\mix}(\p,\q) \le \inf_{n \in \N} \max_{\p \in \Delta(S)} \pi^{\mix}(\p,\q_n) = \inf_{n \in \N} v_n = v.\]

\textbf{Step 4} There exists $\p \in \Delta(S)$ such that $\inf_{\q \in \Delta(T)} \pi^{\mix}(\p,\q) \ge v$.

\noindent \emph{Proof.} We start this step by recalling the notion of weak convergence of probability measures and a useful fact about it. Let $X$ be a metric space, and let $(P_n)_{n=1,2,\ldots}$ be an infinite sequence of Borel probability measures on $X$. We say that the sequence \emph{converges weakly} to a Borel probability measure $P$ on $X$ if $\int_X f \, \mathrm{d}P_n \to \int_X f \, \mathrm{d}P$ for every bounded, continuous function $f: X \to \R$. We will use the following fact: if $X$ is compact then every such sequence has a weakly convergent subsequence (this follows from Prokhorov's theorem; see, e.g., Billingsley~\cite{bill}).

Now, consider for each $n$ an optimal mixed strategy $\p_n \in \Delta(S)$ of player $1$ in $G_n$. We view $\p_n$ as a Borel probability measure on $X = \BB^1_{\down}(G)$, whose support is contained in $\BB^1(G)$. By Step~1, $X$ is a compact metric space. By the fact above, we can find a subsequence $(\p_{n_i})_{i=1,2,\ldots}$ that converges weakly to a Borel probability measure $P$ on $X$. Fixing a pure strategy $t \in T$, we consider the characteristic function $\ind_{[z_t = 1]}: X \to \{0,1\}$ of those $\z = (z_1,z_2,\ldots) \in \BB^1_{\down}(G)$ having $z_t = 1$. This function is clearly continuous with respect to the product topology. Hence, by weak convergence, we have
\begin{equation} \label{eq4}
\int_X \ind_{[z_t=1]} \, \mathrm{d}\p_{n_i} \,\, \to \,\, \int_X \ind_{[z_t=1]} \, \mathrm{d}P.
\end{equation}
These integrals compute the probability, under $\p_{n_i}$ and $P$ respectively, of player $1$ winning against the pure strategy $t$ of player $2$. For $i$ large enough so that $t \in T_{n_i}$, we know that this probability, under $\p_{n_i}$, is at least $v_{n_i}$. As $v_{n_i} \to v$, it follows from~(\ref{eq4}) that the probability under $P$ of player $1$ winning against $t$ is at least $v$. This being true for every pure strategy $t \in T$, it is also true for every mixed strategy $\q \in \Delta(T)$.

The above almost establishes the statement of Step~4, except that $P$ is a Borel probability measure on $\BB^1_{\down}(G)$ rather than on $\BB^1(G)$. To correct this, we replace $P$ by a probability measure $\p = (p_{B_s})_{B_s \in \BB^1(G)}$ on $\BB^1(G)$ as follows. We fix an arbitrary enumeration $B_{s_1},B_{s_2},\ldots$ of the sets in $\BB^1(G)$. We define a mapping $f: \BB^1_{\down}(G) \to \BB^1(G)$ by: $f(A) = B_{s_i}$ where $i$ is the least index so that $A \subseteq B_{s_i}$. Note that $f$ is Borel measurable, because \[ f^{-1}(B_{s_i}) = \{ A \subseteq \N: A \subseteq B_{s_i}\,\, \text{and}\, A \nsubseteq B_{s_j} \,\, \text{for all } j < i\},\] and for each $s \in S$, the set $\{ A \subseteq \N: A \subseteq B_s\}$ is closed in the product topology. This allows us to define the desired $\p$ by $p_{B_s} = P[f^{-1}(B_s)]$ for each $B_s \in \BB^1(G)$. This $\p$ can be viewed as a mixed strategy in $\Delta(S)$. As $A \subseteq B_s$ for every $A \in f^{-1}(B_s)$, player $1$'s winning probability cannot decrease when replacing $P$ by $\p$. Thus, $\p$ satisfies the statement of Step~4.

Steps~3 and 4 together yield the equality~(\ref{eq3}). \end{proof}

A few remarks about Theorem~\ref{thm:22} are in order.

\begin{remark} \label{com:23}
The condition that the family $\BB^1_{\down}(G)$ be ascending-union closed is sufficient but not necessary for the minimax property to hold. As an example, we can extend the LNG game (Example~\ref{ex:13}) by allowing player $2$ to choose $\infty$, which beats all of player $1$'s strategies. After this extension, $\BB^1_{\down}(G)$ is still not ascending-union closed, but the value of the extended game is trivially $0$.
\end{remark}

\begin{remark} \label{com:24}
Requiring that the family $\BB^1(G)$, not $\BB^1_{\down}(G)$, be ascending-union closed is not sufficient for the minimax property to hold. As an example, let $T$ be the disjoint union of two copies $C_1$ and $C_2$ of $\N$. The sets $B_s$ in $\BB^1(G)$ are of one of two forms: either $B_s$ consists of the first $n$ elements of $C_1$ and the $n$-th element of $C_2$ for some $n \in \N$, or vice versa. There are no containments between the sets in $\BB^1(G)$, hence $\BB^1(G)$ is vacuously ascending-union closed. Yet the minimax property does not hold. On the one hand, no mixed strategy $\p$ of player $1$ guarantees a positive winning probability: this is the case in any countable game where all sets $B_s \in \BB^1(G)$ are finite, by essentially the same argument as in Example~\ref{ex:13}. On the other hand, given any mixed strategy $\q$ of player $2$ and $\varepsilon > 0$, player $1$ has a pure strategy $s$ which wins with probability at least $\frac{1}{2} - \varepsilon$. Indeed, we may assume w.l.o.g.\ that $\q$ assigns a probability of at least $\frac{1}{2}$ to $C_1$, hence at least $\frac{1}{2} - \varepsilon$ to the first $n$ elements of $C_1$, for some $n \in \N$. Choosing $s$ such that $B_s$ contains the first $n$ elements of $C_1$ works for player $1$.
\end{remark}

\begin{remark} \label{com:25}
The infimum over $\q \in \Delta(T)$, on either side of equation~(\ref{eq3}), cannot be replaced by a minimum. As an example, consider the following game $G$, an infinite version of matching pennies. Each player chooses a positive integer, and player $1$ wins if and only if they chose the same number. This game is represented by a diagonal matrix $M_G$ with $1$'s on the diagonal. The family $\BB^1_{\down}(G)$ consists of the singletons and the empty set, and is vacuously ascending-union closed. The $\varepsilon$-value of $G$ is clearly $0$, but player $2$ has no exactly optimal mixed strategy, nor an exactly best response to a full-support mixed strategy of player $1$.
\end{remark}

\begin{remark} \label{com:26}
The proof of Step~3 above shows that player $2$ can achieve arbitrarily good approximations of $v$ by using the mixed strategies $\q_n$, which have finite support and w.l.o.g.\ assign rational probabilities to the pure strategies. The same is true of player $1$: arbitrarily good approximations of $v$ can be achieved by starting from the mixed strategy $\p$ obtained in Step~4, and making small enough perturbations which render the support finite and the assigned probabilities rational. However, \emph{exact} optimality for player $1$ cannot always be achieved under these restrictions. By translating Example~4.1 of Aharoni and Holzman~\cite{ah} to the present context, one obtains a game $G$ satisfying the condition of Theorem~\ref{thm:22}, in which player $1$ has a unique optimal mixed strategy $\p$, that has full support and assigns irrational probabilities to some pure strategies. (See Subsection~\ref{sec:43} below on how to translate the terminology of~\cite{ah} to ours.)
\end{remark}

\begin{remark} \label{com:27}
Theorem~\ref{thm:22} may be extended beyond the class of win-lose games, at the cost of losing the combinatorial/set theoretic nature of its condition. Instead of requiring the range of the payoff function $\pi$ to be $\{0,1\}$, we may allow it to be a bounded subset of $\R$, say the interval $[0,1]$. The discrete topology on $\{0,1\}$ is replaced by the standard topology on $[0,1]$. Instead of the combinatorial condition on $\BB^1_{\down}(G)$ in Theorem~\ref{thm:22}, which is used in the proof to show that it is a closed subset of $\{0,1\}^{\N}$, we may require directly that the analog of $\BB^1_{\down}(G)$ be a closed subset of $[0,1]^{\N}$. The rest of the proof works as above, with minor adaptations. We omit the details.
\end{remark}

In the special case when all sets in $\BB^1(G)$ are finite (equivalently, every row of the matrix $M_G$ has finitely many $1$'s), the condition of Theorem~\ref{thm:22} is directly related to the absence of LNG subgames in $G$ (recall Example~\ref{ex:13} and Definition~\ref{def:14}). We show this after introducing the following definition.

\begin{definition} \label{def:28}
Let $G = (S,T,\pi)$ be a game. We say that $G$ is \emph{LNG-free} if no subgame of $G$ is isomorphic to the larger number game (LNG).
\end{definition}

A subgame of $G$ which is isomorphic to LNG appears in the matrix $M_G$ as a submatrix which, possibly after reordering its rows and columns, has the full lower-triangular form shown in Example~\ref{ex:13}. It may be checked that the notion of LNG-freeness is invariant to changes in the tie-breaking rule used in LNG.

\begin{proposition} \label{pro:29}
Let $G = (S,T,\pi)$ be a game such that all sets in the family $\BB^1(G)$ are finite. Then the following three conditions are equivalent:
\begin{enumerate}
\item $\BB^1_{\down}(G)$ is ascending-union closed.
\item $\BB^1_{\down}(G)$ contains no countable ascending chain $A_1 \subset A_2 \subset \cdots \subset A_i \subset \cdots$.
\item $G$ is LNG-free.
\end{enumerate}
\end{proposition}

\begin{proof} (1) $\Rightarrow$ (2) We know that all sets in $\BB^1(G)$, and hence all sets in $\BB^1_{\down}(G)$, are finite. If $\BB^1_{\down}(G)$ does contain a countable ascending chain $A_1 \subset A_2 \subset \cdots \subset A_i \subset \cdots$, then $\cup_{i=1}^{\infty} A_i$ is infinite and therefore not in $\BB^1_{\down}(G)$, so $\BB^1_{\down}(G)$ is not ascending-union closed.

(2) $\Rightarrow$ (3) Suppose that $G' = (S',T',\pi ')$ is a subgame of $G$ which is isomorphic to LNG. Then the sets $\{ t \in T': \pi(s,t) = 1\}$, $s \in S'$, are in $\BB^1_{\down}(G)$ and form a countable ascending chain.

(3) $\Rightarrow$ (1) Supose that $A_1 \subset A_2 \subset \cdots \subset A_i \subset \cdots$ is a countable ascending chain of sets in $\BB^1_{\down}(G)$. We will show that $G$ has a subgame $G' = (S',T',\pi ')$ which is isomorphic to LNG. This requires to construct an infinite sequence of pure strategies, alternating between the two players, of the form $t_1,s_1,t_2,s_2,\ldots,t_k,s_k,\ldots$, in which each strategy beats all the preceding strategies of the opponent. Then, taking $S' = \{s_1,s_2,\ldots\}$ and $T' = \{t_1,t_2,\ldots\}$ we will be done. Our construction of the sequence is inductive, starting with an arbitrary $t_1 \in \cup_{i=1}^{\infty} A_i$ and maintaining the condition that all $t_k$ are in $\cup_{i=1}^{\infty} A_i$. Now suppose that we have constructed the sequence up to a certain point, and we need to find the next strategy to be added. If that strategy is $t_k$, we choose it in $(\cup_{i=1}^{\infty} A_i) \setminus (\cup_{j < k} B_{s_j})$. Note that $\cup_{j < k} B_{s_j}$ is finite by assumption while $\cup_{i=1}^{\infty} A_i$ is infinite, so this choice is feasible. By the definition of the sets $B_{s_j}$, our $t_k$ beats all $s_j$, $j < k$, as required. If the strategy to be added is $s_k$, we look at the set $\{t_1,\ldots,t_k\}$. Each element of this set lies in some $A_i$, and as the $A_i$ form a chain, there exists an $i$ so that $\{t_1,\ldots,t_k\} \subseteq A_i$. Since $A_i \in \BB^1_{\down}(G)$, we can choose $s_k$ so that $A_i \subseteq B_{s_k}$. This $s_k$ beats all $t_j$, $ j \le k$, as required. This completes the inductive construction. \end{proof}

The above proposition leads to the following special case of Theorem~\ref{thm:22}.

\begin{corollary} \label{cor:210}
Let $G = (S,T,\pi)$ be a countable game such that all sets in the family $\BB^1(G)$ are finite. Assume that $G$ satisfies one (hence all) of the equivalent conditions (1)--(3) in Proposition~\ref{pro:29}. Then $G$ satisfies the minimax property in the following form: \[ \max_{\p \in \Delta(S)} \inf_{\q \in \Delta(T)} \pi^{\mix}(\p,\q) = 0 = \inf_{\q \in \Delta(T)} \max_{\p \in \Delta(S)} \pi^{\mix}(\p,\q). \]
\end{corollary}

\begin{proof} This follows immediately from Theorem~\ref{thm:22}, plus the already-made observation that the left-hand side of equality~(\ref{eq3}) equals $0$ when all sets in $\BB^1(G)$ are finite. \end{proof}

The condition in Theorem~\ref{thm:22} focuses on player $1$'s strategies. We can interchange the roles of the players and get a dual result which focuses on player $2$'s strategies. We briefly go through the dual notations.

Let $G = (S,T,\pi)$ be a game. We assume w.l.o.g.\ that the columns in the matrix $M_G$ are all distinct. Thus, each pure strategy $t \in T$ corresponds bijectively to its column $(\pi(s,t))_{s \in S}$ in $\{0,1\}^S$, or equivalently to the subset $B_t$ of $S$ defined by \[B_t = \{ s \in S: \pi(s,t) = 0\}.\] We denote the family of all these sets by $\BB^2(G)$, i.e., \[\BB^2(G) = \{B_t: t \in T\}.\] For example, when $G$ is the larger number game, $\BB^2(G)$ is the family $\{\emptyset, \{1\},\{1,2\},\ldots\}$. We also consider the downward closure $\BB^2_{\down}(G)$ of the family $\BB^2(G)$, namely the family \[\BB^2_{\down}(G) = \{ A \subseteq S: \exists t \in T \,\,\text{such that}\, A \subseteq B_t\}.\] The dual result to Theorem~\ref{thm:22} is the following.

\begin{theorem} \label{thm:211}
Let $G = (S,T,\pi)$ be a countable game, such that the family $\BB^2_{\down}(G)$ is ascending-union closed. Then $G$ satisfies the minimax property in the following form: \[ \sup_{\p \in \Delta(S)} \min_{\q \in \Delta(T)} \pi^{\mix}(\p,\q) = \min_{\q \in \Delta(T)} \sup_{\p \in \Delta(S)} \pi^{\mix}(\p,\q).\]
\end{theorem}

Dualizing Corollary~\ref{cor:210} we get the following.

\begin{corollary} \label{cor:212}
Let $G = (S,T,\pi)$ be a countable game such that all sets in the family $\BB^2(G)$ are finite. Assume that $G$ is LNG-free. Then $G$ satisfies the minimax property in the following form: \[ \sup_{\p \in \Delta(S)} \min_{\q \in \Delta(T)} \pi^{\mix}(\p,\q) = 1 = \min_{\q \in \Delta(T)} \sup_{\p \in \Delta(S)} \pi^{\mix}(\p,\q).\]
\end{corollary}

\section{A characterization of totally minimax games} \label{sec:3}
As explained in the introduction, we should not expect a structural characterization of the games that satisfy the minimax property. A more realistic goal, that we achieve in this section, is a characterization of those games which are totally minimax (Definition~\ref{def:14}). Our result shows that the presence of a copy of LNG (Example~\ref{ex:13}) is the only impediment to being totally minimax. This improves an earlier result of Hanneke, Livni and Moran~\cite{hlm}, who proved the same under the additional condition that the VC-dimension of the game be finite. This condition requires the existence of a finite bound on the size of subsets $T' \subseteq T$ so that for every subset $A \subseteq T'$ there exists $s \in S$ for which $\{t \in T': \pi(s,t) = 1\} = A$.

\begin{theorem} \label{thm:31}
Let $G = (S,T,\pi)$ be a game. Then $G$ is totally minimax if and only if $G$ is LNG-free.
\end{theorem}

\begin{proof} As LNG violates the minimax property, the ``only if" direction is trivial. In the opposite direction, we prove below that if $G$ is LNG-free then $G$ satisfies the minimax property. This will suffice, because all subgames of an LNG-free game are themselves LNG-free.

From now on, we assume that $G$ is LNG-free. For most of the proof, we also assume that $G$ is countable (in the last step we will show how to reduce the result for arbitrary games to the countable case). Our basic approach is the same as in the proof of Theorem~\ref{thm:22}. Recall that $\BB^1(G)$ is the family of sets $B_s$, $s \in S$, which correspond to the rows of the matrix $M_G$. We treat $\BB^1(G)$ as a subset of the compact metric space $\{0,1\}^{\N}$. In the proof of Theorem~\ref{thm:22} we extended $\BB^1(G)$ to $\BB^1_{\down}(G)$, which was shown, under the condition of that theorem, to be a closed subset of $\{0,1\}^{\N}$. Here, instead, we extend $\BB^1(G)$ to its closure in $\{0,1\}^{\N}$, which we denote $\cl(\BB^1(G))$. Thus, $\cl(\BB^1(G))$ is itself a compact metric space.

We want to show that $G$ satisfies the minimax property. By Corollaries~\ref{cor:210} and \ref{cor:212}, we already know that this follows from LNG-freeness in the special case when $\BB^1(G)$ or $\BB^2(G)$ consists of finite sets only. We use this fact in the first step of our proof, where we show that for any $\varepsilon > 0$, every pure strategy of player $1$ in the extended game corresponding to $\cl(\BB^1(G))$ can be replaced by a mixed strategy in the original game $G$, while losing at most $\varepsilon$ against any pure strategy of player $2$.

\textbf{Step 1} Let $\varepsilon > 0$, and let $\z = (z_1,z_2,\ldots) \in \cl(\BB^1(G))$. Then there exists $\p \in \Delta(S)$ such that for all $t \in T$:
\begin{equation} \label{eq5}
z_t = 1 \,\, \Longrightarrow \,\, \pi^{\mix}(\p,t) \ge 1 - \varepsilon.
\end{equation}
Moreover, the mixed strategy $\p = (p_s)_{s \in S}$ may be chosen so that its support is finite and all probabilities $p_s$, $s \in S$, are rational.

\noindent \emph{Proof.} As $\z \in \cl(\BB^1(G))$, we can find a sequence $\z^{(i)} = (z_1^{(i)}, z_2^{(i)},\ldots) \in \BB^1(G)$, $i = 1,2,\ldots$, which converges to $\z$ in the product topology (i.e., pointwise). We consider the subgame $G' = (S',T',\pi ')$ of $G$, where the strategies $s \in S'$ are those corresponding to $\z^{(1)},\z^{(2)},\ldots$, and the strategies $t \in T'$ are those for which $z_t = 1$. If there exists $s \in S'$ such that $\pi(s,t) = 1$ for all $t \in T'$, then we can take $\p$ to be the pure strategy $s$. Otherwise, both $S'$ and $T'$ must be infinite, so $G'$ is a countable game. Since $(\z^{(i)})_{i=1,2,\ldots}$ converges pointwise to $\z$, for every $t \in T'$, all but finitely many $i$ satisfy $z_t^{(i)} = z_t = 1$. This means that all sets in $\BB^2(G')$ are finite. Thus, $G'$ satisfies the assumptions of Corollary~\ref{cor:212} (it is LNG-free as a subgame of $G$). Therefore its $\varepsilon$-value is $1$, so we can find a mixed strategy $\p$ of player $1$ (in $G'$, hence also in $G$) so that $\pi^{\mix}(\p,t) \ge 1 - \varepsilon$ for all $t \in T'$. Such $\p$ satisfies~(\ref{eq5}), and moreover can be chosen to have finite support and assign rational probabilities (see Remark~\ref{com:26}).

\textbf{Step 2} Let $G_n = (S,T_n,\pi_n)$ be the subgame of $G$ obtained by restricting player $2$ to the first $n$ pure strategies. Then for every $n$, the game $G_n$ has a value $v_n$, and the sequence $(v_n)_{n=1,2,\ldots}$ is non-increasing and converges to some $v \in [0,1]$.

\noindent \emph{Proof.} This is proved just like the corresponding step in the proof of Theorem~\ref{thm:22}.

We go on to show that the $v$ of Step~2 is the $\varepsilon$-value of $G$.

\textbf{Step 3} We have $\inf_{\q \in \Delta(T)} \sup_{\p \in \Delta(S)} \pi^{\mix}(\p.\q) \le v$.

\noindent \emph{Proof.} This is again similar to the second part of the corresponding step in the proof of Theorem~\ref{thm:22}. Except for having a supremum, not a maximum over $\p \in \Delta(S)$, the proof is the same.

\textbf{Step 4} We have $\sup_{\p \in \Delta(S)} \inf_{\q \in \Delta(T)} \pi^{\mix}(\p,\q) \ge v$.

\noindent \emph{Proof.} Repeating the arguments in the corresponding step in the proof of Theorem~\ref{thm:22}, with $X = \cl(\BB^1(G))$ instead of $\BB^1_{\down}(G)$, we find a Borel probability measure $P$ on $\cl(\BB^1(G))$, under which player $1$ wins with probability at least $v$ against every strategy of player $2$. We show below, using Step~1, how to replace $P$ by a mixed strategy $\p \in \Delta(S)$, while incurring an arbitrarily small loss in this guarantee.

Let $\varepsilon > 0$ be given. We denote by $\Delta^0_{\Q}(S)$ the set of all mixed strategies in $\Delta(S)$ that have a finite support and assign only rational probabilities (this set is countable). We fix an arbitrary enumeration $\p^{(1)},\p^{(2)},\ldots$ of $\Delta^0_{\Q}(S)$. We define a mapping $f: \cl(\BB^1(G)) \to \Delta^0_{\Q}(S)$ by: $f(\z) = \p^{(i)}$ where $i$ is the least index so that $\p^{(i)}$ satisfies condition~(\ref{eq5}) with respect to $\varepsilon$ and $\z = (z_1,z_2,\ldots)$. By Step~1, this is well defined. The mapping $f$ is Borel measurable, because $f^{-1}(\p^{(i)})$ consists of those $\z = (z_1,z_2,\ldots) \in \cl(\BB^1(G))$ that satisfy: $z_t = 0$ for all $t \in T$ such that $\pi^{\mix}(\p^{(i)},t) < 1 - \varepsilon$ (this defines a closed set in the product topology) and the negation of this for every $j < i$ (this defines an open set).

Consider the probability distribution $\p$ over $S$ defined by the following two-step random process. First, $\z = (z_1,z_2,\ldots)$ is chosen in $\cl(\BB^1(G))$ according to the Borel probability measure $P$ obtained above. Next, a pure strategy $s$ is chosen in $S$ according to the mixed strategy $f(\z)$. We fix $t \in T$ and lower bound the probability that $\pi(s,t) = 1$, when $s$ is chosen according to $\p$ as explained. The known property of $P$ ensures that in the first step $z_t = 1$ with probability at least $v$. By~(\ref{eq5}), the conditional probability that $\pi(s,t) = 1$ given that $z_t = 1$ is at least $1 - \varepsilon$. Thus, the probability that $\pi(s,t) = 1$ is at least $v(1 - \varepsilon)$. This is true for every pure strategy $t \in T$, hence also for every mixed strategy $\q \in \Delta(T)$. So, our $\p \in \Delta(S)$ satisfies $\inf_{\q \in \Delta(T)} \pi^{\mix}(\p,\q) \ge v(1 - \varepsilon)$. As $\varepsilon > 0$ was arbitrary, the claim of Step~4 is proved.

Steps~3 and 4 together show that $G$ satisfies the minimax property, completing the proof in the countable case. In the next step, we extend the result from the countable case to the general case.

\textbf{Step 5} Let $G = (S,T,\pi)$ be a game that violates the minimax property. Then $G$ has a countable subgame $G' = (S',T',\pi ')$ that violates the minimax property.

\noindent \emph{Proof.} Here we use a construction from~\cite{hlm}. We denote the two sides of equation~(\ref{eq2}) by $\ubar{v}$ and $\bar{v}$, respectively, and assume that \[ \sup_{\p \in \Delta(S)} \inf_{\q \in \Delta(T)} \pi^{\mix}(\p,\q) = \ubar{v} < \bar{v} = \inf_{\q \in \Delta(T)} \sup_{\p \in \Delta(S)} \pi^{\mix}(\p,\q). \] We construct an infinite sequence of mixed strategies with finite support, alternating between the two players, of the form $\p^{(1)},\q^{(1)},\p^{(2)},\q^{(2)},\ldots,\p^{(k)},\q^{(k)},\ldots$, in which each strategy gives the player using it an expected payoff at least as good as the better among $\ubar{v}$ and $\bar{v}$, against any pure strategy in the union of the supports of the preceding strategies of the opponent. Namely, for every $k$ we require:
\begin{equation} \label{eq6}
\pi^{\mix}(\p^{(k)},t) \ge \bar{v} \,\,\text{for all}\,\, t \in \cup_{j < k} \mathrm{supp}(\q^{(j)}),
\end{equation}
\begin{equation} \label{eq7}
\pi^{\mix}(s,\q^{(k)}) \le \ubar{v} \,\,\text{for all}\,\, s \in \cup_{j \le k} \mathrm{supp}(\p^{(j)}).
\end{equation}

The sequence is constructed inductively, starting with an arbitrary $\p^{(1)} \in \Delta(S)$ with finite support. Suppose that we have constructed the sequence up to a certain point, and we need to find the next strategy to be added. If that strategy is $\p^{(k)}$, we consider the subgame $G^{(k)} = (S, T^{(k)},\pi^{(k)})$ obtained by restricting player $2$ to the pure strategies $t \in \cup_{j < k} \mathrm{supp}(\q^{(j)})$. This subgame is equivalent to a finite game, because the $\q^{(j)}$ have finite support. Hence $G^{(k)}$ has a value $v^{(k)}$, and clearly $v^{(k)} \ge \bar{v}$. Taking $\p^{(k)}$ to be an optimal mixed strategy (with finite support) of player $1$ in $G^{(k)}$ satisfies condition~(\ref{eq6}). If the strategy to be added is $\q^{(k)}$, a similar argument allows us to choose it so as to satisfy condition~(\ref{eq7}). This completes the inductive construction.

Now, consider the subgame $G' = (S',T',\pi ')$ where $S' = \cup_{j=1}^{\infty} \mathrm{supp}(\p^{(j)})$ and $T' = \cup_{j=1}^{\infty} \mathrm{supp}(\q^{(j)})$. Clearly, $S'$ and $T'$ are at most countable. We claim that
\begin{equation} \label{eq8}
\sup_{\p \in \Delta(S')} \inf_{\q \in \Delta(T')} \pi^{\mix}(\p,\q) \le \ubar{v} < \bar{v} \le \inf_{\q \in \Delta(T')} \sup_{\p \in \Delta(S')} \pi^{\mix}(\p,\q).
\end{equation}
This will show that $G'$ also violates the minimax property, thereby completing the proof. We prove the first inequality in~(\ref{eq8}), the proof of the last inequality is similar. It suffices to show that for any fixed $\p \in \Delta(S')$ we have $\inf_{k \in \N} \pi^{\mix}(\p,\q^{(k)}) \le \ubar{v}$. This follows from~(\ref{eq7}) and the fact that when $s \in S'$ is chosen according to $\p$, the probability that $s \in \cup_{j \le k} \mathrm{supp}(\p^{(j)})$ tends to $1$ as $k \to \infty$. \end{proof}

Theorem~\ref{thm:31} says that LNG-freeness is a necessary and sufficient condition for being totally minimax. However, checking whether a given game satisfies this condition may not be easy. The following corollary lists some simple-to-check counting conditions, each of which precludes an LNG subgame and therefore suffices for being totally minimax.

\begin{corollary} \label{cor:32}
Let $G = (S,T,\pi)$ be a game. Assume that the representing matrix $M_G$ satisfies one (or more) of the following condititons:
\begin{enumerate}
\item The number of $0$ entries in each row (except maybe finitely many) is finite.
\item The number of $1$ entries in each row (except maybe finitely many) is finitely bounded.
\item The number of $0$ entries in each column (except maybe finitely many) is finitely bounded.
\item The number of $1$ entries in each column (except maybe finitely many) is finite.
\end{enumerate}
Then $G$ is totally minimax.
\end{corollary}

\begin{remark} \label{com:33}
Steve Hanneke~\cite{per} has extended Theorem~\ref{thm:31} to characterize those zero-sum games with a bounded payoff function which are totally minimax. Instead of our single forbidden binary submatrix that represents LNG, he forbids all real-valued submatrices of the same shape in which the supremum of the entries above the diagonal is strictly less than the infimum of the entries below the diagonal. His proof extends the techniques used here to the real case, and also applies Theorem~\ref{thm:31} in the context of a non-trivial discretization argument. A similar extension to the real case of the above-mentioned result of Hanneke, Livni and Moran~\cite{hlm} was obtained by Daskalakis and Golowich~\cite{dg}.
\end{remark}

\section{Connections and applications} \label{sec:4}
\subsection{Enforcing the convergence of nonnegative number series} \label{sec:41}
Here we address the question: Which families of finite sets of positive integers enforce the convergence of nonnegative number series, in the sense of the following definition?
\begin{definition} \label{def:41}
Let $\FF$ be a family of finite subsets of $\N$. We say that $\FF$ is \emph{convergence enforcing} if the following implication holds for every infinite sequence $(a_i)_{i=1,2,\ldots}$ of nonnegative real numbers:
\begin{equation} \label{eq9}
\sum_{i \in A} a_i \le 1 \,\,\text{for all}\, A \in \FF \,\, \Longrightarrow \,\, \sum_{i \in \N} a_i < \infty.
\end{equation}
\end{definition}

The standard definition of convergence of a series yields the simplest example of a convergence enforcing family: when $\FF$ is the family of all finite initial segments of $\N$, the implication~(\ref{eq9}) holds with $\sum_{i \in \N} a_i \le 1$. Another example is the family of all finite initial segments of the even positive integers or of the odd ones, which enforces $\sum_{i \in \N} a_i \le 2$. There are many other, more sophisticated examples. Here is a preliminary observation.

\begin{proposition} \label{pro:42}
Let $\FF$ be a family of finite subsets of $\N$ such that $\cup \FF = \N$. Assume that $\FF$ is convergence enforcing. Then there exists a finite constant $C = C(\FF) < \infty$ so that the following implication holds for every infinite sequence $(a_i)_{i=1,2,\ldots}$ of nonnegative real numbers:
\begin{equation} \label{eq10}
\sum_{i \in A} a_i \le 1 \,\,\text{for all}\,\, A \in \FF \,\, \Longrightarrow \,\, \sum_{i \in \N} a_i \le C.
\end{equation}
\end{proposition}

\begin{proof} Suppose, for the sake of contradiction, that no such constant $C$ exists. Applying this to all $C$ of the form $2^n$, $n \in \N$, we can find for every $n$ a sequence $(a_i^{(n)})_{i=1,2,\ldots}$ as above, such that $\sum_{i \in A} a_i^{(n)} \le 1$ for all $A \in \FF$, yet $\sum_{i \in \N} a_i^{(n)} > 2^n$. Note that $a_i^{(n)} \le 1$ for all $n$ and $i$, because every $i$ belongs to some $A \in \FF$. Hence $a_i = \sum_{n \in \N} \frac{a_i^{(n)}}{2^n}$ is a finite nonnegative number for each $i$. However, the sequence $(a_i)_{i=1,2,\ldots}$ violates~(\ref{eq9}): we have \[ \sum_{i \in A} a_i = \sum_{i \in A} \sum_{n \in \N} \frac{a_i^{(n)}}{2^n} = \sum_{n \in \N} \frac{1}{2^n} \sum_{i \in A} a_i^{(n)} \le \sum_{n \in \N} \frac{1}{2^n} = 1 \] for all $A \in \FF$, yet \[ \sum_{i \in \N} a_i = \sum_{i \in \N}
\sum_{n \in \N} \frac{a_i^{(n)}}{2^n} = \sum_{n \in \N} \frac{1}{2^n} \sum_{i \in \N} a_i^{(n)} \ge \sum_{n \in \N} 1 = \infty. \] This contradicts our assumption that $\FF$ is convergence enforcing. \end{proof}

Another observation is that $\FF$ is convergence enforcing if and only if its downward closure is. Therefore, when studying this property we may assume that the family $\FF$ is downward closed, i.e., $A \subseteq B$ and $B \in \FF \Rightarrow A \in \FF$.

While downward-closed convergence enforcing families may look very different from one another, we show here that they all have one thing in common: they must contain an infinite ascending chain. This is an application of Corollary~\ref{cor:210}.

\begin{theorem} \label{thm:43}
Let $\FF$ be a family of finite subsets of $\N$. Assume that $\FF$ is downward closed and convergence enforcing. Then $\FF$ contains an infinite ascending chain $A_1 \subset A_2 \subset \cdots \subset A_i \subset \cdots$.
\end{theorem}

\begin{proof} Suppose, for the sake of contradiction, that $\FF$ contains no infinite ascending chain.

We may assume that $\cup \FF = \N$. Indeed, the set $\N \setminus (\cup \FF)$ must be finite, otherwise convergence enforcing clearly fails. If $\N \setminus (\cup \FF)$ is non-empty, we may consider $\cup \FF$ as our copy of $\N$.

Assuming that $\cup \FF = \N$, we can find a constant $C = C(\FF) < \infty$ so that~(\ref{eq10}) holds for every infinite sequence $(a_i)_{i=1,2,\ldots}$ of nonnegative real numbers. Consider the game $G = (S,T,\pi)$ where $T = \N$ and $\BB^1(G) = \FF$. As $\FF$ is downward closed, we have $\BB^1_{\down}(G) = \FF$ as well. By our assumption, this family contains no infinite ascending chain. Thus, condition~(2) in Proposition~\ref{pro:29} holds. So Corollary~\ref{cor:210} applies, showing that $\inf_{\q \in \Delta(T)} \max_{\p \in \Delta(S)} \pi^{\mix}(\p,\q) = 0$. Let $0 < \varepsilon < \frac{1}{C}$, and let $\q = (q_i)_{i \in \N}$ be a mixed strategy of player $2$ such that $\pi^{\mix}(s,\q) \le \varepsilon$ for every pure strategy $s$ of player $1$. Now consider the sequence $(a_i)_{i=1,2,\ldots}$ where $a_i = \frac{q_i}{\varepsilon}$. We get a contradiction by showing that this sequence violates~(\ref{eq10}). To check the premise of~(\ref{eq10}), let $A \in \FF$ and let $s$ be the corresponding pure strategy of player $1$. Then $\sum_{i \in A} a_i = \frac{1}{\varepsilon} \sum_{i \in A} q_i = \frac{1}{\varepsilon} \pi^{\mix}(s,\q) \le 1$. On the other hand, $\sum_{i \in \N} a_i = \frac{1}{\varepsilon} > C$. \end{proof}

The contrapositive of the theorem says that if a downward-closed family $\FF$ of finite subsets of $\N$ contains no infinite ascending chain, then there is a divergent series $\sum_{i \in \N} a_i$ that ``fools" $\FF$ in the sense that $\sum_{i \in A} a_i \le 1$ for all $A \in \FF$. A nice example is the family $\FF$ of all sets $A$ such that $|A| \le \min A$, which is fooled by the harmonic series. But there are richer, more complicated examples of families without an infinite ascending chain, for which it is hard to explicitly find a fooling series.

\subsection{Separation of row and column densities in infinite binary matrices} \label{sec:42}
The notions of upper/lower asymptotic density of a subset of $\N$, or equivalently of an infinite binary sequence, play an important role in number theory. For a sequence $\mathbf{a} = (a_i)_{i \in \N} \in \{0,1\}^{\N}$, they are respectively defined by: \[ \bar{d}(\mathbf{a}) = \limsup_{n \to \infty} \frac{1}{n} \sum_{i=1}^n a_i \,\,\, \text{and} \,\,\, \ubar{d}(\mathbf{a}) = \liminf_{n \to \infty} \frac{1}{n} \sum_{i=1}^n a_i. \] If the two are equal, their common value $d(\mathbf{a})$ is called the asymptotic density of $\mathbf{a}$.

Here we look at an $\N$ by $\N$ binary matrix $A = (a_{ij})_{i \in \N, j \in \N} \in \{0,1\}^{\N \times \N}$. We denote its $i$-th row by $\rr^{(i)} = (a_{ij})_{j \in \N}$ and its $j$-th column by $\cc^{(j)} = (a_{ij})_{i \in \N}$. How are the row and column upper/lower asymptotic densities related?

\begin{definition} \label{def:44}
Let $A$ be an $\N$ by $\N$ binary matrix. We say that the row and column upper/lower asymptotic densities are \emph{separated} if there exist real numbers $0 \le \alpha < \beta \le 1$ such that:
\begin{equation} \label{eq11}
\bar{d}(\rr^{(i)}) \le \alpha \,\,\text{and}\,\, \ubar{d}(\cc^{(j)}) \ge \beta \,\,\text{for all}\,\, i,j \in \N.
\end{equation}
\end{definition}

Clearly, the analog of~(\ref{eq11}) would be impossible in a finite matrix. But such separation can occur in an $\N$ by $\N$ matrix. Indeed, in the binary matrix from Example~\ref{ex:13} (in which $a_{ij} = 1 \Leftrightarrow i \ge j$), every row has density $0$ while every column has density $1$. We call this matrix fully lower-triangular. We say that a matrix $A$ contains a fully lower-triangular submatrix, if it has a submatrix which, possibly after reordering rows and columns, is fully lower-triangular. Here we show, as an application of Theorem~\ref{thm:31}, that separation can only occur in the presence of a fully lower-triangular submatrix.

\begin{theorem} \label{thm:45}
Let $A$ be an $\N$ by $\N$ binary matrix in which the row and column upper/lower asymptotic densities are separated. Then $A$ contains a fully lower-triangular submatrix.
\end{theorem}

\begin{proof} Consider the game $G = (S,T,\pi)$ represented by the matrix $A$ (here $S = T = \N$). We claim that
\begin{equation} \label{eq12}
\sup_{\p \in \Delta(S)} \inf_{\q \in \Delta(T)} \pi^{\mix}(\p,\q) \le \alpha < \beta \le \inf_{\q \in \Delta(T)} \sup_{\p \in \Delta(S)} \pi^{\mix}(\p,\q),
\end{equation}
where $\alpha$ and $\beta$ witness the separation condition~(\ref{eq11}) for $A$. This will show that $G$ violates the minimax property. By Theorem~\ref{thm:31}, it will follow that $G$ has an LNG subgame, or equivalently, $A$ contains a fully lower-triangular submatrix, completing the proof.

We prove the first inequality in~(\ref{eq12}), the proof of the last inequality is similar. Let $\p = (p_i)_{i \in S} \in \Delta(S)$ and $\varepsilon > 0$ be given. We have to find some $\q \in \Delta(T)$ such that $\pi^{\mix}(\p,\q) \le \alpha + \varepsilon$. First, we find a large enough $k$ so that $\sum_{i=1}^k p_i \ge 1 - \frac{\varepsilon}{2}$. Next, for each $i = 1,\ldots,k$, we look at the corresponding row $\rr^{(i)}$ in the matrix $A$. By~(\ref{eq11}) its upper density is at most $\alpha$, so there exists $n_i \in \N$ such that $\frac{1}{n} \sum_{j=1}^n a_{ij} \le \alpha + \frac{\varepsilon}{2}$ for all $n \ge n_i$. Now let $n = \max\{n_1,\ldots,n_k\}$ and let $\q \in \Delta(T)$ be the uniform distribution over the first $n$ columns. Then we have \[ \pi^{\mix}(\p,\q) = \sum_{i=1}^{\infty} \sum_{j=1}^n \frac{p_i}{n} a_{ij} = \sum_{i=1}^k \frac{p_i}{n} \sum_{j=1}^n a_{ij} + \sum_{i=k+1}^{\infty} \frac{p_i}{n} \sum_{j=1}^n a_{ij} \le \bigl(\alpha + \frac{\varepsilon}{2}\bigr) + \frac{\varepsilon}{2} = \alpha + \varepsilon, \] as desired. \end{proof}

\subsection{Duality of fractional matchings and covers in infinite hypergraphs} \label{sec:43}

A hypergraph is a pair $H = (V,\EE)$, where $V$ is a set of vertices and $\EE$ is a set of subsets of $V$, called edges. The sets $V$ and $\EE$, and the edges $E$ in $\EE$, may be finite or infinite. A fractional matching in $H$ is a function $f: \EE \to \R_+$ such that $\sum_{E \ni v} f(E) \le 1$ for every vertex $v \in V$. Its size is $|f| = \sum_{E \in \EE} f(E)$. A fractional cover in $H$ is a function $g: V \to \R_+$ such that $\sum_{v \in E} g(v) \ge 1$ for every edge $E \in \EE$. Its size is $|g| = \sum_{v \in V} g(v)$. Two important hypergraph invariants are \[ \nu^*(H) = \sup\{ |f| : f \,\,\text{is a fractional matching in}\,\, H\}\] and \[ \tau^*(H) = \inf\{ |g| : g \,\,\text{is a fractional cover in}\,\, H\}.\] Both invariants may be finite or infinite.

Fractional matchings and covers in finite hypergraphs have numerous applications in combinatorics and discrete optimization. A fundamental fact is the equality $\nu^*(H) = \tau^*(H)$ which holds for every finite hypergraph by the duality theorem of linear programming. What can be said for infinite hypergraphs?

\begin{definition} \label{def:46}
We say that the hypergraph $H$ satisfies the \emph{weak duality property} if $\nu^*(H) = \tau^*(H)$.
\end{definition}

A simple example of an infinite hypergraph where weak duality fails is the \emph{tail hypergraph} $H = (\N,\EE)$, where the edges in $\EE$ are all sets of the form $E_j = \{ i \in \N: i \ge j\}$, $j \in \N$. Here $\nu^*(H) =1$ while $\tau^*(H) = \infty$.

A partial sub-hypergraph of $H = (V,\EE)$ is obtained by restricting $V$ and $\EE$ to some subsets: it is of the form $H' = (V',\EE')$ where $V' \subseteq V$ and $\EE' = \{ E \cap V': E \in \EE_0\}$ for some $\EE_0 \subseteq \EE$. As an application of Theorem~\ref{thm:31}, we characterize here those hypergraphs all of whose partial sub-hypergraphs satisfy the weak duality property.

\begin{theorem} \label{thm:47}
Let $H = (V,\EE)$ be a hypergraph. Then every partial sub-hypergraph $H'$ of $H$ satisfies the weak duality property if and only if $H$ has no partial sub-hypergraph $H'$ which is isomorphic to the tail hypergraph.
\end{theorem}

\begin{proof} Given a hypergraph $H = (V,\EE)$, we associate with it a game $G = G(H) = (V,\EE,\pi)$ in which player $1$ chooses a vertex $v \in V$, player $2$ chooses an edge $E \in \EE$, and player $1$ wins if and only if $v \in E$. Note that if $H$ is the tail hypergraph then $G(H)$ is isomorphic to the LNG game. We claim that
\begin{equation} \label{eq13}
\sup_{\p \in \Delta(V)} \inf_{\q \in \Delta(\EE)} \pi^{\mix}(\p,\q) = \frac{1}{\tau^*(H)} \,\,\, \text{and} \,\, \inf_{\q \in \Delta(\EE)} \sup_{\p \in \Delta(V)} \pi^{\mix}(\p,\q) = \frac{1}{\nu^*(H)}.
\end{equation}
This will show that $H$ satisfies the weak duality property if and only if $G(H)$ satisfies the minimax property. Clearly, partial sub-hypergraphs correspond to subgames, so the present theorem is in fact a reformulation of Theorem~\ref{thm:31}.

We prove the first part of~(\ref{eq13}), the proof of the second part is similar. If $g$ is a fractional cover in $H$ of finite size $|g|$, then the mixed strategy $\p = (p_v)_{v \in V}$ in which $p_v = \frac{g(v)}{|g|}$ guarantees to player $1$ in $G(H)$ an expected payoff of at least $\frac{1}{|g|}$. Conversely, if $\p = (p_v)_{v \in V}$ is a mixed strategy of player $1$ in $G(H)$ that guarantees an expected payoff of at least $\alpha > 0$, then the function $g : V \to \R_+$ where $g(v) = \frac{p_v}{\alpha}$ is a fractional cover in $H$ of size $|g| = \frac{1}{\alpha}$. These two reciprocal maps together establish the first part of~(\ref{eq13}). \end{proof}

We note that Aharoni and Holzman~\cite{ah} studied the duality of fractional matchings and covers in infinite hypergraphs. They introduced two versions of duality: the weak one that we use here, and a stronger one which is an infinite counterpart of the complementary slackness conditions of linear programming duality. The main positive results in~\cite{ah} gave sufficient conditions for weak duality, which essentially correspond to conditions~(2) and (4) in Corollary~\ref{cor:32}. However, no characterization result corresponding to Theorem~\ref{thm:31} was obtained there.

\subsection{A characterization of stable theories in first-order logic} \label{sec:44}
Stability is a central concept in model theory, with many applications to algebra and other areas of mathematics. Shelah~\cite{sh} established several equivalent definitions of stability, of which we give here one that is related to our work.

\begin{definition} \label{def:48}
Let $\LL$ be a language in first-order logic, and let $T$ be a complete theory in this language. Let $\varphi(x_1,\ldots,x_n;y_1,\ldots,y_m)$ be a formula in $\LL$ with $n+m$ free variables. We say that $\varphi$ has the \emph{order property} over $T$ if there exists a model $\MM$ of $T$, and there exist two infinite sequences $(\mathbf{a}_i)_{i=1,2,\ldots}$ and $(\bb_j)_{j=1,2,\ldots}$ where each $\mathbf{a}_i$ is an $n$-tuple of elements of $\MM$ and each $\bb_j$ is an $m$-tuple of elements of $\MM$, such that \[ \MM \models \varphi(\mathbf{a}_i;\bb_j) \,\, \Longleftrightarrow \,\, i \ge j. \] The theory $T$ is \emph{stable} if no such formula $\varphi$ has the order property over $T$.
\end{definition}

Here we apply Theorem~\ref{thm:31} to give a game theoretic characterization of stable theories. Fix a language $\LL$ and a theory $T$ as above. Given a formula $\varphi(x_1,\ldots,x_n;y_1,\ldots,y_m)$ in $\LL$ and a model $\MM$ of $T$, we associate with them a game $G = G(\LL,T,\varphi,\MM)$ as follows: player $1$ chooses an $n$-tuple $\mathbf{a}$ of elements of $\MM$, player $2$ chooses an $m$-tuple $\bb$ of elements of $\MM$, and player $1$ wins if and only if $\MM \models \varphi(\mathbf{a};\bb)$.

\begin{theorem} \label{thm:49}
Let $\LL$ be a language in first-order logic, and let $T$ be a complete theory in $\LL$. Then $T$ is stable if and only if for every formula $\varphi(x_1,\ldots,x_n;y_1,\ldots,y_m)$ in $\LL$ and every model $\MM$ of $T$ the associated game $G = G(\LL,T,\varphi,\MM)$ is totally minimax.
\end{theorem}

\begin{proof} Fixing $\LL,T,\varphi$ and $\MM$, we note that a pair of sequences $(\mathbf{a}_i)_{i=1,2,\ldots}$ of $n$-tuples of elements of $\MM$ and $(\bb_j)_{j=1,2,\ldots}$ of $m$-tuples of elements of $\MM$ such that $\MM \models \varphi(\mathbf{a}_i;\bb_j) \Leftrightarrow i \ge j$ corresponds to an LNG subgame of $G = G(\LL,T,\varphi,\MM)$. Thus, $\varphi$ has the order property over $T$ if and only if there exists a model $\MM$ of $T$ such that $G = G(\LL,T,\varphi,\MM)$ has an LNG subgame. Therefore, by the definition above, $T$ is stable if and only if for every $\varphi$ and $\MM$ the game $G = G(\LL,T,\varphi,\MM)$ is LNG-free, or equivalently (by Theorem~\ref{thm:31}), $G$ is totally minimax. \end{proof}

We remark that in the case of a countable language $\LL$, it is known that it suffices to check the definition above for countable models $\MM$ of $T$. In this case, the games $G$ in Theorem~\ref{thm:49} are all countable as well.

We note that there are several other examples where strategic games nicely capture concepts in logic and set theory. But we are not aware of such instances where mixed strategies are involved, as is the case here.

\subsection{Online learning with simple  predictors} \label{sec:45}
We briefly describe here a standard model of online learning. There is a set $\XX$ of instances, each of which may be labelled by $0$ or $1$. A concept class $\mathbb{H} \subseteq \{0,1\}^{\XX}$ is given; each hypothesis $h \in \mathbb{H}$ represents a possible labelling of all instances. A learner, who knows the class $\mathbb{H}$ but not the ``true" hypothesis, is faced with a sequence of instances that are revealed in online fashion. In each round, the learner tries to guess the label of the given instance, and then learns whether the guess was correct or not, before proceeding to the next round. A fundamental result of Littlestone~\cite{lit} characterizes those concept classes $\mathbb{H}$ that are online learnable, in the sense that the learner can keep the total number of mistakes below some finite constant. To this end, he defined a combinatorial complexity measure called the Littlestone dimension of $\mathbb{H}$, and showed that $\mathbb{H}$ is online learnable if and only if its Littlestone dimension is finite.

Hanneke, Livni and Moran~\cite{hlm} showed recently that when $\mathbb{H}$ is online learnable, the learning may be achieved by a ``simple predictor" of a certain particular type. Their analysis used the following game $G = G(\mathbb{H})$ associated with the concept class $\mathbb{H}$: the learner chooses a hypothesis $h \in \mathbb{H}$, an adversary chooses a labelled instance $(x,y) \in \XX \times \{0,1\}$, and the learner wins if and only if $h(x)=y$. A key step in their proof is showing that when $\mathbb{H}$ has finite Littlestone dimension, the game $G(\mathbb{H})$ satisfies the minimax property. This motivated them to prove, more generally, that a game $G$ of finite VC-dimension (see explanation in Section~\ref{sec:3}) is totally minimax if and only if it is LNG-free. Our Theorem~\ref{thm:31} generalizes their result by removing the assumption of finite VC-dimension. This has no consequence for their original application, though, because the VC-dimension is bounded above by the Littlestone dimension.

Incidentally, the argument in~\cite{hlm} involved yet another complexity measure from learning theory: the threshold dimension, which is related to LNG-freeness. In a previous paper, Alon, Livni, Malliaris and Moran~\cite{almm} had proved that the threshold dimension is finite if and only if the Littlestone dimension is (this was essentially a learning theory reformulation of a result of Shelah~\cite{sh} about stability -- see Subsection~\ref{sec:44}). Hanneke, Livni and Moran~\cite{hlm} used this to show that if $\mathbb{H}$ has finite Littlestone dimension then $G(\mathbb{H})$ is LNG-free, allowing them to deduce the minimax property in their application.

\vspace{10pt}

\textbf{Acknowledgments} The research for this paper originated in discussions with Shay Moran, aiming to generalize his result with Hanneke and Livni. I am grateful to him and to Ron Aharoni, Zachary Chase, Bogdan Chornomaz, Yuval Dagan, Yuval Filmus and Steve Hanneke for many helpful discussions. In particular, Yuval Dagan suggested the argument in the last step of the proof of Theorem~\ref{thm:31}; Steve Hanneke communicated his result mentioned in Remark~\ref{com:33}; the application presented in Subsection~\ref{sec:42} came up in discussions among some of the above-mentioned people.

\bibliographystyle{plain}
\bibliography{minimax}

\end{document}